\documentclass[conference]{IEEEtran}

\IEEEoverridecommandlockouts

\ifCLASSINFOpdf
\else
\fi
\usepackage{graphicx}
\DeclareGraphicsExtensions{.pdf,.png,.jpg, .tif}
\usepackage{amsmath}
\usepackage{latexsym}
\usepackage{graphicx,color}
\usepackage{bm}
\usepackage{amssymb}
\usepackage{array}
\usepackage{setspace}
\usepackage{fancyhdr}
\usepackage{float}
\usepackage{algorithm}
\usepackage{algorithmicx}
\usepackage{algpseudocode} 
\newcommand{\subparagraph}{}
\usepackage{titlesec}
\usepackage[stable]{footmisc}
\usepackage{balance}
\usepackage[table,xcdraw]{xcolor}
\usepackage{csquotes}
\MakeOuterQuote{"}
\makeatletter
\newcommand{\vast}{\bBigg@{2.6}}
\newcommand{\Vast}{\bBigg@{4.6}}
\newcommand{\ignore}[1]{}
\makeatother

\newtheorem{theorem}{Theorem}

\usepackage{parskip}

\usepackage{cite}
\usepackage{balance}

\begin{document}

\title{Adaptive Task Allocation for Mobile Edge Learning}

\author{\IEEEauthorblockN{Umair Mohammad and Sameh Sorour}
\IEEEauthorblockA{Department of Elecrical and Computer Engineering, University of Idaho, Moscow, ID, USA\\
moha2139@vandals.uidaho.edu, samehsorour@uidaho.edu}
}

\maketitle

\begin{abstract}
This paper aims to establish a new optimization paradigm\ignore{ for jointly allocating resources and distributed learning tasks to} to efficiently execute distributed learning tasks on wireless edge nodes with heterogeneous computing and communication capacities. We will refer to this new paradigm as ``Mobile Edge Learning (MEL)''. The problem of adaptive task allocation for MEL is considered in this paper with the aim to maximize the learning accuracy, while guaranteeing that the total times of data distribution/aggregation over heterogeneous channels, and local computation on heterogeneous nodes, are bounded by a preset duration. The problem is first formulated as a quadratically-constrained integer linear problem. Being NP-hard, the paper relaxes it into a non-convex problem over real variables. We then propose a solution based on deriving analytical upper bounds on the optimal solution of this relaxed problem using KKT conditions. The merits of this proposed solution is exhibited by comparing its performances to both numerical approaches and the equal task allocation approach.  
\end{abstract}

\IEEEpeerreviewmaketitle

\section{Introduction}
\label{Section1_Introduction}
The accelerated migration towards the era of smart cities mandates the deployment of a large number of Internet-of-Things (IoT) devices, generating exponentially increasing amounts of data at the edge of the network.\ignore{ This data is currently sent to cloud severs for analytics and decision-making that improve the performance of a wide range of systems and services. However,} It is expected that the rate and nature of this generated data will prohibit their centralized processing and analytics at cloud servers. Indeed, the size of data will surpass the capabilities of current and even future wireless networks and internet backbone to transfer them to cloud data-centers \cite{Ref1_Original}. In addition,\ignore{ the nature of this data and} the time-criticality of their analytics will enforce 90\% of their processing to be done locally at edge servers or even at the edge (mostly mobile) nodes themselves (e.g., smart phones,\ignore{ laptops,} monitoring cams, drones, connected\ignore{and autonomous} vehicles) \cite{Ref2_Original}.

This booming need for edge processing is supported by the late advancements in mobile edge computing (MEC) \cite{P0_Letaief_Survey,RefsPap_MEC126, Liu2017a, RefsPap_MEC121} and hierarchical MEC (H-MEC) \cite{RefsPap_MEC125,Moha1812:Multi}. While the former enables edge nodes to\ignore{ decide on whether to} either perform their computing task locally or offload them to the edge servers, the latter further enables task offloading among edge nodes themselves. Such task offloading decisions are usually made while respecting the heterogeneous computing and communication capacities of the edge nodes and links, respectively, so as to minimize\ignore{ task completion} delays and/or energy consumption. While most of the previous works on MEC/H-MEC were limited to\ignore{ scenarios with nodes performing} independent computing tasks, one of the most important forthcoming application of MEC/H-MEC is implementing distributed learning (DL) on edge nodes. Using DL, multiple edge nodes can collectively analyze and learn from their possessed data with no or limited dependence on cloud/edge servers. 

Distributed learning (DL) has recently attracted much attention within the machine learning (ML) realm, motivated by two trending scenarios in the cyber-world, namely the distributed-datasets scenario and the task-parallelization scenario. In the former, (big) datasets\ignore{ about the same phenomenon} are separately generated/collected by multiple nodes, but cannot be transferred to a central hub for analytics due to some constraints (e.g., bandwidth, privacy) \cite{HarvardDistributed}. In this case, the learning process cycles between these nodes (a.k.a. learners) performing local training/learning on their individual datasets, and a central orchestrator collecting the locally derived parameters, performing global processing, and returning globally updated parameters to the learners. On the other hand, the task-parallelization scenario usually involves a main node that parallelizes the learning process over its local dataset on multiple cores/nodes due to one or multiple reasons (e.g., limited main node resources, faster processing, lower energy consumption) \cite{GoogleDistWired}. Thus, this main node must distribute its dataset (or most commonly random parts of it) to these multiple cores/nodes for local learning, and itself acts as the orchestrator of the process.

Obviously, performing learning in mobile, edge, and IoT environments are the clear manifestation of both above scenarios. Indeed, edge nodes generate/collect data that they cannot share to edge/cloud servers due to bandwidth limitations. Many of them are also computationally-limited, and can thus use H-MEC to parallelize the learning process over multiple neighboring nodes for faster processing and/or lower energy consumption. Yet, most works on DL are studied over wired and/or non-heterogeneous distributed computing and data transfer environments \cite{GoogleDistWired}. Extending these DL studies to resource-constrained edge servers and nodes was not explored until very recently. Some works \cite{Tuor_01_AdaptiveControl, Tuor_02_2018_DemoAbstract} aimed to unify the number of local learning iterations in resource-constrained edge environments in order to maximize the learning accuracy. The proposed approach jointly optimized the number of local learning and global update cycles at the learners and the orchestrator, respectively. However, these works overlooked the inherent heterogeneities in the computing and communication capacities of different learners and links, respectively. The implications of such heterogeneities on optimizing the task allocation to different learners, selecting learning models, improving learning accuracy, minimizing local and global cycle times, and/or minimizing energy consumption, are clearly game-changing, but yet never investigated.

To the best of the authors' knowledge, this work is the first attempt to synergize the novel trends of DL and H-MEC, in order to establish an optimization framework for efficiently executing distributed learning tasks on a neighboring set of heterogeneous wireless edge nodes. We will refer to this new paradigm as ``Mobile Edge Learning (MEL)''. This paper will inaugurate this MEL research, by considering the problem of adaptive task allocation for distributed learning over heterogeneous wireless edge learners (i.e., edge nodes with heterogeneous computing capabilities and heterogeneous wireless links to the orchestrator). This task allocation will be conducted so as to maximize the learning accuracy, while guaranteeing that the total times of data distribution/processing, and parameter aggregation in such a heterogeneous environment are bounded by a preset duration by the orchestrator. The maximization of the learning accuracy is achieved by maximizing the number of local learning iterations per global update cycle \cite{Tuor_01_AdaptiveControl}. 

To this end, the problem is first formulated as quadratically-constrained integer linear problem. Being an NP-hard problem, the paper relaxes it to a non-convex problem over real variables. Analytical upper bounds on the optimal solution of this relaxed problem are derived using KKT conditions. The proposed algorithm will thus start from these computed bounds, and then runs suggest-and-improve steps to reach a feasible integer solution. The merits of this proposed solution will be exhibited through extensive simulations, comparing its performances to both numerical solutions and the equal task allocation approach of \cite{Tuor_01_AdaptiveControl,Tuor_02_2018_DemoAbstract}.


\section{System Model for MEL}
\label{Section02__SystemModelParameters}

\subsection{Distributed Learning Preliminaries}
 Distributed learning is defined as the operation of running one machine learning task over a system of $K$ learners. 
 Due to processing/memory and learning time constraints resulting from large sizes of typical datasets, a stochastic gradient descend (SGD) approach is often employed for DL with mini-batch training\ignore{ in the MEL realm}\footnote{This choice is also justified by the generality of this approach and its proven superior learning accuracy performance compared to deterministic gradient decent (GD) methods \cite{SGDisBetter}.}.
 In SGD, the learning cycles are performed on randomly picked subsets (that are referred to as "batches") from the multiple datasets locally stored at the different learners (i.e., distributed-datasets scenario) or the one global orchestrator dataset (i.e., task-parallelization scenario). Clearly, the latter scenario fully encompasses the former, but only adds to it the batch transfer component from the orchestrator to the learners. We will thus mainly consider the latter scenario in this paper, but will show the variations in the model\ignore{/formulation/solution} when the former scenario is considered.
 
 Learner $k$, $k \in\{ 1, 2, \dots, K\}$ trains its local learning model or learns from a batch of size $d_k$ data samples, so as to minimize the local loss function \cite{Tuor_01_AdaptiveControl}. The total size of all batches is denoted by $d=\sum_{k=1}^K d_k$, which is usually preset by the orchestrator $O$ given its computational capabilities, the desired accuracy, and the time-constraints of the training/learning process. 
 The number of local iterations (a.k.a. local updates) run by learners on their allocated batch is denoted by $\tau$. 
 
 
 Once each learner finishes its $\tau$ local updates, it forwards the resulting local parameter matrix to the orchestrator. The size of each learner's parameter matrix usually depends on its assigned batch size, its employed learning model, and the selected data precision. The orchestrator then aggregates the local update matrices and updates the unique global parameter matrix by minimizing the global loss function. This step from the orchestrator is called the global update process. Afterwards, the global parameter matrix is sent to learners and the process cycles between local update cycles and global update processes until the orchestrator decides to end it once it reaches a target accuracy. Interested readers in the local/global loss function minimization and local parameter aggregations are referred to \cite{Tuor_01_AdaptiveControl,Tuor_02_2018_DemoAbstract} for further details. 

 In a synchronous setting, the global update process should occur in periodic cycles, that we will refer to as the global update cycles. This duration should include the transmission of the batches and/or the global parameter matrix to learners, the time for $\tau$ local update cycles at each of them, their forwarding of the local parameter matrices to the orchestrator, and the global update process.
 The orchestrator performs the aggregation of the parameters only once after all learners send back their result in that global update cycle.
 
 To this end, define the time $T$ as the duration needed to perform the first three steps (i.e., excluding the global update process). We will refer to the time $T$ as the global cycle clock to differentiate it from the total global cycle duration consisting of the global cycle clock and the global update process time. We care about this time $T$ as its included processes fully depend on the batch sizes and employed learning model sizes at the learners, and most importantly on how those relate to the computing capabilities of the different learners as well as their communication capacities to the orchestrator. Whereas these components may not be of significant influence when DL is executed over controlled wired and infrastructural servers, their high heterogeneity can tremendously impact the performance of DL when applied in wireless and mobile edge environments. This is where the MEL paradigm comes into play as will be detailed in the next section.

\subsection{Transition to MEL}
In this section, we introduce the parameters that relate to the heterogeneity of the computing and communication capacities of wireless edge learners, and how they relate to the steps of the global update clock duration. Define $B_k^{data}$ as the number of bits of the batch allocated to learner $k$, which can be expressed as follows:
\begin{equation}
B_k^{data} = d_k\mathcal{F}\mathcal{P}_d
\end{equation}
where $\mathcal{F}$ is the number of features in the dataset, and $\mathcal{P}_d$ is the data bit precision. On the other hand, the number of bits of learner $k$'s local parameter matrix (denoted by $B_k^{model}$) can be expressed for learner $k$ as:
\begin{equation}
B_k^{model} = \mathcal{P}_m\left(d_kS_d+S_m\right)
\end{equation}
where $\mathcal{P}_m$ is the model bit precision (typically floating point precision). As shown in the above equation, the local parameter matrix size consists of two parts, one depending on the batch size (represented by the term $d_kS_d$, where $S_d$ is the number of model coefficients related to each sample of the batch), and the other related to the constant size of the employed ML model (denoted by $S_m$). 


The orchestrator sends the  bit concatenation of the data batch and initial global parameter matrix to learner $k$ with power $P_{ko}$ over a wireless channel, having a bandwidth $W$ and a \ignore{complex} channel power gain $h_{ko}$.
Once this information is received by each learner $k$, it sets is local parameter matrix $\tilde{\mathbf{w}}_k$ to the initial matrix $\mathbf{w}$ provided by the orchestrator. It then performs $\tau$ local update cycles on this local parameter matrix, using its allocated batch. Typically in ML, the algorithm sequentially goes over all features of each data sample once in one iteration (or epoch).
Consequently, the number\ignore{$X_k$} of computations required per iteration $X_k$ is equal to:
\begin{equation}
    X_k = d_k C_m
\end{equation}
which clearly depends on the number of data samples $d_k$ assigned to each learner and the computational complexity $C_m$ of the model.\ignore{ Although the data may be stored as other types, the operations themselves are typically floating point operations. }

We assume that the bits of the computed local parameter matrix $\tilde{\mathbf{w}}_k$ at each learner $k$ is transmitted back to the orchestrator with the same power $P_{ko}$ on the same channel. The orchestrator will then re-compute the global parameter matrix $\mathbf{w}$ as described in \cite{Tuor_01_AdaptiveControl}. Once computed, it send this matrix back with a new random batch of samples from the dataset to each learner, and the process repeats.

Given the above description, the times of each learner $k$, $\forall~k$, whose sum must be bounded by the global update clock $T$, can be detailed as follows. The first time $t_k^S$ is the time needed to send the allocated batch and global parameter matrix $\mathbf{w}$ to learner $k$. $t_k^S$ can thus be expressed as (where $N_0$ is the noise power spectral density)\footnote{Note that, for the distributed-datasets scenario, the only difference in the model is that the first term of the numerator will not exist.}:
    \begin{equation}
    t_{k}^S \ignore{&= \dfrac{B_k^{data}+B_k^{model}}{R_{k}} \nonumber\\
    &} = \dfrac{d_k\mathcal{F}\mathcal{P}_d + \mathcal{P}_m \left(d_k\mathcal{S}_d+\mathcal{S}_m\right)}{W\log_2\left(1+\frac{P_{ko} h_{ko}}{N_0}\right)}
    \label{eq:edge_time_sending}
    \end{equation}
The second time consists of $\tau$ times $t_k^C$, needed by learner $k$ to perform one local update cycle. Defining $f_k$ as learner $k$'s local processor frequency dedicated to the DL task, $t_k^C$ can be expressed as:
    \begin{equation}
    t_k^C = \dfrac{X_k}{f_k} = \frac{d_k C_m}{f_k}
    \label{Eq_9_localTimeExecution}
    \end{equation}
The third and last time $t_k^R$ is the one needed for learner $k$ to send its updated local parameter matrix $\tilde{\mathbf{w}}_k$ to the orchestrator. Assuming the channels are reciprocal and does not change during the duration of one global update cycle, $t_k^R$ can be computed as:
    \begin{equation}
     t_{k}^R = \ignore{\dfrac{B_k^{model}}{R_{k}} = } \dfrac{\mathcal{P}_m \left(d_k\mathcal{S}_d+\mathcal{S}_m\right)}{W\log_2\left(1+\frac{P_{ko} h_{ko}}{N_0}\right)}
    \label{eq:edge_time_receiving}
    \end{equation}

Thus, the total time $t_k$ taken by learner $k$ to complete the above three processes is equal to:
\begin{align}
\label{Eq_12_timeForLearnerk}
    t_k &= t_{k}^S+\tau t_k^C + t_k^R   \nonumber\\
    & = \dfrac{d_k\mathcal{F}\mathcal{P}_d + 2\mathcal{P}_m \left(d_k\mathcal{S}_d+\mathcal{S}_m\right)}{W\log_2\left(1+\frac{P_{ko} h_{ko}}{N_0}\right)} + \tau \dfrac{d_k\mathcal{C}_m}{f_k} 
\end{align}
\ignore{We will refer to the time $t_k$ of learner $k$ as its round-trip distributed processing duration.}As described above, $t_k\leq T$ $\forall~k$ for the orchestrator to have all the needed information to perform its global cycle update process. 

\section{Problem Formulation}
\ignore{As mentioned in Section I, the objective of this paper is to optimize the task allocation (i.e., distributed batch sizes $d_k$ for each learner $k$), so as to maximize the accuracy of the distributed learning process in each global cycle (and thus eventually the accuracy of the entire learning process) within a preset global cycle clock $T$ by the orchestrator.} It is well established in the literature that the loss function in general GD/SGD-based ML are minimized (and thus the learning accuracy is maximized) by increasing the number of learning iterations \cite{GoogleDistWired}. For DL, this is equivalent to maximizing the number of local iterations $\tau$ in each global cycle \cite{StalenessAwarePaper}. Thus, maximizing the MEL accuracy is achieved by maximizing $\tau$ \ignore{for a fixed number of global update cycles}. Given the above model and facts, this paper's objective, as stated in Section I, can be re-worded as optimizing the assigned batch sizes $d_k$ to each of the learners so as to maximize the number of local iterations $\tau$ per global updated cycle, while bounding $t_k$ $\forall~k$ by the preset global cycle clock $T$. The optimization variables in this problems are thus $\tau$ and $d_k$, $k\in\{1,\dots,K\}$. 

We can thus re-write the expression of $t_k$ in (\ref{Eq_12_timeForLearnerk}) as a function of the optimization variables as follows:
\begin{equation}
    t_k = C_k^2  \tau d_k + C_k^1 d_k + C_k^0
    \label{Eq_16_C1isCompact}
\end{equation}
where $C_k^2$, $C_k^1$, and $C_k^1$ represent the quadratic, linear, and constant coefficients of learner $k$ in terms of the optimization variables $\tau$ and $d_k$, expressed as:\footnote{Note that, for the distributed-datasets scenario, the only difference in the model is that the first term of the numerator in (\ref{eq:linear}) will not exist.}
\begin{align}
C_k^2 &= \dfrac{\mathcal{C}_m}{f_k}\\
C_k^1 &= \dfrac{\mathcal{F}\mathcal{P}_d+2\mathcal{P}_m\mathcal{S}_d}{W\log_2\left(1+\frac{P_{ko} h_{ko}}{N_0}\right)} \label{eq:linear}\\
C_k^0 &= \dfrac{2\mathcal{P}_m\mathcal{S}_m}{W\log_2\left(1+\frac{P_{ko} h_{ko}}{N_0}\right)}
\end{align}
Clearly the relationship between $t_k$ and the optimization variables $d_k$ and $\tau$ is quadratic. Furthermore, the optimization variables $\tau$ and $d_k$ $\forall~k$ are all non-negative integers. Consequently, the problem of interest in this paper can be formulated as an integer linear program with quadratic and linear constraints as follows: \footnote{Note that, for the distributed-datasets scenario, the only difference in the formulation is the simpler expression of $C^1_k$. Thus, the problem type and solution remain the same with different $C^1_k$ expressions for the two scenarios.}
\begin{subequations}
\begin{align}
    &\qquad\operatornamewithlimits{max}_{\tau,{d}_k~\forall~k}  \quad \tau\\
    & \quad \nonumber\\
    \text{s.t. }\qquad & C_k^2  \tau d_k + C_k^1 d_k + C_k^0 \leq T, \quad k = 1,\ldots,K \label{orignial-time-const}\\
    & \sum_{k = 1}^{K}d_k = d \label{orignial-batch-const}\\ 
    & \tau \in \mathcal{Z}_+ \label{orignial-tau-const}\\
    & d_k \in \mathcal{Z}_+, \quad k = 1,\ldots,K \label{orignial-d-const}
    \end{align}
    \label{Eq_13_OurProb}
\end{subequations}
Constraint (\ref{orignial-time-const}) guarantees that $t_k\leq T$ $\forall~k$. Constraint (\ref{orignial-batch-const}) ensures that the sum of batch sizes assigned to all learners is equal to the total dataset size that the orchestrator needs to analyze. Constraints (\ref{orignial-tau-const}) and (\ref{orignial-d-const}) are simply non-negativity and integer constraints for the optimization variables. Note that the solutions of (\ref{Eq_13_OurProb}) having $\tau$ and/or all $d_k$'s being zero represent settings where MEL is not feasible.

Thus, the above problem is an integer linear program with quadratic constraints (ILPQC), which is well-known to be NP-hard \cite{QIP_NP_ArXiV}. We will thus proposed a simpler solution to it through relaxation of the integer constraint in the next section.
 
\section{Proposed Solution}\label{Section3_Solution}

\subsection{Problem Relaxation}
As shown in the previous section, the problem of interest in this paper is NP-hard due to its integer decision variables. We thus propose to simplify the problem by relaxing the integer constraints in (\ref{orignial-tau-const}) and (\ref{orignial-d-const}), solving the relaxed problem, then rounding the obtained real results back into integers. The relaxed problem can be thus given by:
\begin{subequations}
\begin{align}
    &\qquad\operatornamewithlimits{max}_{\tau,{d}_k~\forall~k}  \quad \tau\\
    & \quad \nonumber \\
    \text{s.t. }\qquad & C_k^2  \tau d_k + C_k^1 d_k + C_k^0 \leq T, \quad k = 1,\ldots,K \label{relaxed-time-const}\\
    & \sum_{k = 1}^{K}d_k = d \label{relaxed-batch-const}\\ 
    & \tau \geq 0 \label{relaxed-tau-const}\\
    & d_k \geq 0, \quad k = 1,\ldots,K \label{relaxed-d-const}
    \end{align}
    \label{Eq_13_RelaxedProblem}
\end{subequations}
The above resulting program becomes a linear program with quadratic constraints. This problem can be solved by using interior-point or ADMM methods, and there are efficient solvers (such as OPTI) that implement these approaches \cite{OPTI_CW12a}. From the analytical viewpoint, the associated matrices to each of the quadratic constraints in (\ref{relaxed-time-const}) can be written in a symmetric form. However, these matrices will have two non-zero values that are positive and equal\ignore{ to each other}. The eigenvalues will thus sum to zero, which means these matrices are not positive semi-definite, and hence the relaxed problem is not convex. Consequently, we cannot derive the optimal solution of this problem analytically. Yet, we can still derive upper bounds on the optimal variables and solution using KKT conditions. 

\subsection{Upper Bounds using KKT conditions}
The Lagrangian function of the relaxed problem is given by:
\begin{multline}\label{lagrangian}
    L\left(\mathbf{x}, \lambda, \nu, \alpha\right) = -\tau + \sum_{k = 1}^{K} \lambda_k\left(C_k^2  \tau d_k + C_k^1 d_k + C_k^0 - T\right) + \\ \nu_1\left(\sum_{k = 1}^{K}d_k - d\right) - \nu_2\left(\sum_{k = 1}^{K}d_k - d\right) - \alpha_0\tau - \sum_{k = 1}^{K} \alpha_kd_k    
\end{multline}
where the $\lambda_k$'s $k\in\{1,\dots,K\}$, $\nu_1$/$\nu_2$, and $\alpha_0$/$\alpha_k$ $k\in\{1,\dots,K\}$, are the Lagrangian multipliers associated with the time constraints of the $K$ learners in (\ref{relaxed-time-const}), the total batch size constraint in (\ref{relaxed-batch-const}), and the non-negative constraints of all the optimization variables in (\ref{relaxed-tau-const}) and (\ref{relaxed-d-const}), respectively. Note that the equality constraint in (\ref{relaxed-batch-const}) can be represented in the form of two inequality constraints $\sum_{k = 1}^{K}d_k \leq d$ and $-\sum_{k = 1}^{K}d_k = d) \leq 0$, and thus is associated with two Lagrangian multipliers $\nu_1$ and $\nu_2$. 
 
Using\ignore{ the well-known} KKT conditions, the following theorem introduces upper bounds on the optimal variables of the relaxed problem.
\begin{theorem}
\ignore{The optimal values }$d_k^*$ \ignore{of the allocated batch sizes to different users in the relaxed problem }satisfy the following bound:
\begin{equation}
    d_k^* \leq \dfrac{T-C_k^0}{\tau C_k^2+C_k^1} \qquad \forall~k
    \label{Eq_30_BoundOnDk}
\end{equation}
Moreover, the analytical upper bound on $\tau$ belongs to the solution set of the polynomial given by: 
\begin{equation}
    d\prod_{k=1}^{K} \left(\tau^*+b_k\right) - \sum_{k=1}^{K} a_k \prod_{\substack{l=1 \\ l\neq k}}^{K} \left(\tau^*+b_l\right) = 0
\label{Eq_36_AnalBound}    
\end{equation}
where $r_k^0 = C_k^0 - T$, $a_k = -\dfrac{r_k^0}{C_k^2}$ and $b_k = \dfrac{C_k^1}{C_k^2}$.
\label{theorem1}
\end{theorem}
\begin{proof}
The proof of this theorem can be found in Appendix A.
\end{proof}

Though it was expected that the above bounds\ignore{ for $d_k^*$ $\forall~k$ in (\ref{Eq_30_BoundOnDk}) and $\tau^*$ from solving (\ref{Eq_36_AnalBound})} should undergo suggest-and-improve (SAI) steps to reach a feasible solution, we have found in our extensive simulations presented in Section V that these expressions were always already feasible. This means that no SAI steps were needed, and that the $d_k^*$ expressions can be directly used\ignore{ for batch allocation} to achieve the optimal $\tau^*$ for the relaxed problem.

\section{Simulation Results}
In this section, the proposed\ignore{ analytical-based} solution from the derived results in (\ref{Eq_30_BoundOnDk}) and (\ref{Eq_36_AnalBound}) (UB-Analytical) is tested in MEL scenarios emulating realistic\ignore{ learning and} edge node environments. We also show the merits of the proposed solution compared to the equal task allocation (ETA) scheme employed in \cite{Tuor_01_AdaptiveControl} as well as the OPTI-based numerical solution to the relaxed problem in (\ref{Eq_13_RelaxedProblem}).\ignore{ We will first introduce the simulation setting, then present the testing results.}

\subsection{Simulation Environment, Dataset, and Learning Model}
\ignore{In our simulation, }The considered edge nodes in the simulation are assumed to be located in an area of 50m of radius. Half of these nodes emulates the capacity of typical fixed/portable computing devices (e.g., laptops, tablets, road-side units) and the other half emulates the capacity of commercial micro-controllers (e.g., Raspberry Pi) that can be attached to different indoor or outdoor systems (e.g., smart meters, traffic cameras). The setting thus emulates an edge environment that can be located either indoor or outdoor. The employed channel model\ignore{ between these devices} is summarized in Table \ref{Table_1_OfParameters}, which emulates 802.11-type links between the edge nodes. 

We test our proposed scheme with the well-known MNIST \cite{MNIST_IEEE} dataset. This dataset consists of 60,000 images 28x28 images (784 features). The employed ML model for this data is a 3-layer neural network with the following configuration $[784, 300, 124, 60, 10]$. For this model, the weight matrix $\mathbf{w}\ignore{ = \left[w_1, w_2\right]}$ is the set\ignore{concatenation } of four sub-matrices, where $w_1 \text{ is } 784\times300 \text{, } w_2 \text{ is } 300 \times 124 \text{, } w_3 \text{ is } 124 \times 60  \text{ and } w_4 \text{ is } 60\times10$\ignore{, neither of which depending on the batch size ($S_d = 0$)}. Thus, the size of the model is 8,974,080 bits, which is fixed for all edge nodes. The forward and backward passes will require 1,123,736 floating point operations \cite{NNComplexity}.

\begin{table}[!t]
\centering
\small
\caption{List of simulation parameters}
\begin{tabular}{|l|l|}
\hline
Parameter                                 & Value                         \\ \hline
Attenuation Model                           & $7+2.1\log(R)$ dB \cite{WiFiChannelModel}         \\ \hline
System Bandwidth $B$                        & 100 MHz                              \\ \hline
Node Bandwidth $W$                          & 5 MHz                               \\ \hline
Device proximity $R$                        & 50m                                 \\ \hline
Transmission Power $P_k$                   & 23 dBm                              \\ \hline
Noise Power Density $N_0$                  & -174 dBm/Hz                         \\ \hline
Computation Capability $f_k$             & 2.4 GHz and 700 MHz             \\ \hline
MNIST Dataset size $d$                 & 60,000 images                        \\ \hline
MNIST Dataset Features $\mathcal{F}$  & 784 ($~28 \times 28~$) pixels     \\ \hline
\end{tabular}
\label{Table_1_OfParameters}
\end{table}
 
\subsection{Number of Local Updates}
\begin{figure}[t]
   \centering
   \includegraphics[scale=0.58]{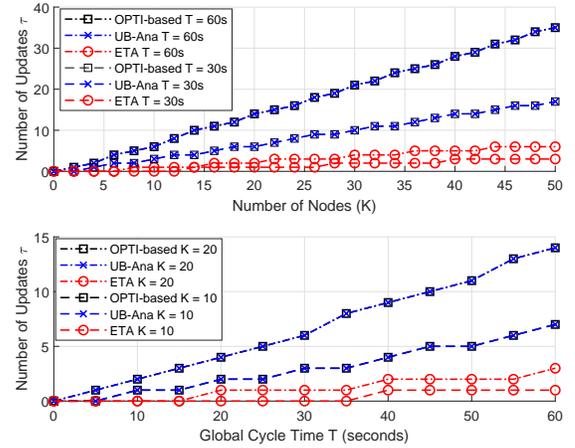}
   \caption{Achievable number of local update cycles by all schemes (a) vs $K$ for $T=30$s and $60$s (b) vs $T$ for $K=10$ and $20$.}
   \label{figure2label}
\end{figure}
\begin{figure}[t]
   \centering
   \includegraphics[scale=0.58]{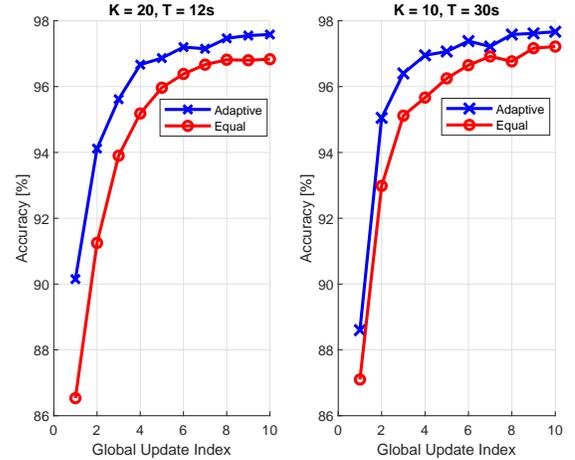}
   \caption{Learning accuracy progression after global cycles updates for $K=20$}
   \label{figure3label}
\end{figure}
Fig. \ref{figure2label} shows the results for training using the MNIST dataset with the aforementioned deep neural network model. The upper and lower sub-figures depict the number of local iterations $\tau$ achieved by all tested approaches versus the number of edge nodes (for $T=30$s and $T=60$s) and against the global cycle time (for $K=10$ and $K=20$), respectively. We can first notice from both figures that $\tau$ increases as the number of edge nodes increases. This is trivial as less batch sizes or more time will be allocated to each of the nodes as their number or the global cycle increase, respectively. We can also see that the performance of the UB-Analytical, and OPTI-numerical solutions are identical for all simulated number of edge nodes and global update clocks. 

We can finally observe from both sub-figures that the optimized scheme achieves a significantly larger number of local updates compared to the ETA scheme. For instance, the optimized scheme makes it possible to perform 6 updates (as opposed to 1) for 20 nodes with a cycle time of $30$s, a gain of 600\%. When $K = 10$, at $T= 60$s the OPTI-based approach for adaptive batch allocation gives $\tau = 7$ updates whereas only 3 updates are possible with ETA, a gain of 233\%. Another interesting result is that the performance of ETA scheme for $T=60$s or $K=20$ is actually much lower than the performance of our proposed solutions for $T=30$s or $K=10$ in the upper and lower sub-figures respectively. In other words, our scheme can achieve a better level of accuracy as the ETA scheme in half the time or half the number of nodes.

\subsection{Learning Accuracy}
The left and right sub-figures in Fig. \ref{figure3label} depict the progression of learning accuracy achieved by both, our adaptive scheme and ETA, right after global update cycles of $T=12$s each for $K=20$ and $T=30$s each for $K=10$. The figure shows a significant improvement achieved by our proposed scheme over ETA in reaching high accuracy in less number of global cycles. For $K=20$ and $T=12$s, to reach an accuracy of 96.67\%, the adaptive scheme require 4 cycles as opposed to 7 for ETA, a reduction of about 43\% (i.e., 48s). Furthermore, an accuracy of 97\% can be achieved with our scheme which is not possible with ETA. For the case of $K=10$ and $T=30$s, the optimized scheme can cross the 97\% mark for accuracy in 5 updates whereas ETA requires 9 updates which represents a reduction in time of 2 minutes or about 56\%. This clearly exhibits the merits of the adaptive scheme in reaching learning goals in much less time and number of cycles, which also saves significantly on nodes energy.

\section{Conclusion}
This paper inaugurates the research efforts towards establishing the novel MEL paradigm, enabling the design of optimized distributed learning solutions on wireless edge nodes with heterogeneous computing and communication capacities. As a first MEL problem of interest, the paper focused on exploring the adaptive task allocation solutions that would maximize the number of local learning iterations on distributed learners (and thus improving the learning accuracy), while abiding by the global cycle clock of the orchestrator. The problem was formulated as an NP-hard ILPQC problem, which was then relaxed into a non-convex problem over real variables. Analytical upper bounds on the relaxed problem's optimal solution were then derived and were found to solve it optimally in all simulated scenarios. Through extensive simulations using the well-known MNIST dataset, the proposed analytical solution was shown to both achieve the same performance as the numerical solvers of the ILPQC problem, and to significantly outperform the equal ETA approach.

\balance
\bibliographystyle{IEEEtran}
\bibliography{181EdgeLearning}

\begin{thebibliography}{10}
\providecommand{\url}[1]{#1}
\csname url@samestyle\endcsname
\providecommand{\newblock}{\relax}
\providecommand{\bibinfo}[2]{#2}
\providecommand{\BIBentrySTDinterwordspacing}{\spaceskip=0pt\relax}
\providecommand{\BIBentryALTinterwordstretchfactor}{4}
\providecommand{\BIBentryALTinterwordspacing}{\spaceskip=\fontdimen2\font plus
\BIBentryALTinterwordstretchfactor\fontdimen3\font minus
  \fontdimen4\font\relax}
\providecommand{\BIBforeignlanguage}[2]{{%
\expandafter\ifx\csname l@#1\endcsname\relax
\typeout{** WARNING: IEEEtran.bst: No hyphenation pattern has been}%
\typeout{** loaded for the language `#1'. Using the pattern for}%
\typeout{** the default language instead.}%
\else
\language=\csname l@#1\endcsname
\fi
#2}}
\providecommand{\BIBdecl}{\relax}
\BIBdecl

\bibitem{Ref1_Original}
\BIBentryALTinterwordspacing
M.~Chiang and T.~Zhang, ``{Fog and IoT: An Overview of Research
  Opportunities},'' \emph{IEEE Internet of Things Journal}, vol.~3, no.~6, pp.
  854--864, dec 2016. [Online]. Available:
  \url{http://ieeexplore.ieee.org/document/7498684/}
\BIBentrySTDinterwordspacing

\bibitem{Ref2_Original}
\BIBentryALTinterwordspacing
{Rhea Kelly}, ``{Internet of Things Data To Top 1.6 Zettabytes by 2020 --
  Campus Technology},'' 2015. [Online]. Available:
  \url{https://campustechnology.com/articles/2015/04/15/internet-of-things-data-to-top-1-6-zettabytes-by-2020.aspx}
\BIBentrySTDinterwordspacing

\bibitem{P0_Letaief_Survey}
\BIBentryALTinterwordspacing
Y.~Mao, C.~You, J.~Zhang, K.~Huang, and K.~B. Letaief, ``{A Survey on Mobile
  Edge Computing: The Communication Perspective},'' \emph{IEEE Communications
  Surveys {\&} Tutorials}, vol.~19, no.~4, pp. 2322--2358, 2017. [Online].
  Available: \url{http://arxiv.org/abs/1701.01090}
\BIBentrySTDinterwordspacing

\bibitem{RefsPap_MEC126}
\BIBentryALTinterwordspacing
C.~You and K.~Huang, ``{Mobile Cooperative Computing: Energy-Efficient
  Peer-to-Peer Computation Offloading},'' pp. 1--33, 2017. [Online]. Available:
  \url{http://arxiv.org/abs/1704.04595}
\BIBentrySTDinterwordspacing

\bibitem{Liu2017a}
M.~Liu and Y.~Liu, ``{Price-Based Distributed Offloading for Mobile-Edge
  Computing with Computation Capacity Constraints},'' \emph{IEEE Wireless
  Communications Letters}, pp. 1--4, 2017.

\bibitem{RefsPap_MEC121}
\BIBentryALTinterwordspacing
Y.~Li, L.~Sun, and W.~Wang, ``{Exploring device-to-device communication for
  mobile cloud computing},'' \emph{2014 IEEE International Conference on
  Communications (ICC)}, pp. 2239 -- 44, 2014. [Online]. Available:
  \url{http://dx.doi.org/10.1109/ICC.2014.6883656}
\BIBentrySTDinterwordspacing

\bibitem{RefsPap_MEC125}
\BIBentryALTinterwordspacing
X.~Cao, F.~Wang, J.~Xu, R.~Zhang, and S.~Cui, ``{Joint Computation and
  Communication Cooperation for Mobile Edge Computing},'' in \emph{16th
  International Symposium on Modeling and Optimization in Mobile, Ad Hoc, and
  Wireless Networks (WiOpt)}, 2018. [Online]. Available:
  \url{http://arxiv.org/abs/1704.06777}
\BIBentrySTDinterwordspacing

\bibitem{Moha1812:Multi}
U.~Y. Mohammad and S.~Sorour, ``{Multi-Objective Resource Optimization for
  Hierarchical Mobile Edge Computing},'' in \emph{2018 IEEE Global
  Communications Conference: Mobile and Wireless Networks (Globecom2018 MWN)},
  Abu Dhabi, United Arab Emirates, dec 2018.

\bibitem{HarvardDistributed}
S.~Teerapittayanon, B.~McDanel, and H.~T. Kung, ``{Distributed Deep Neural
  Networks over the Cloud, the Edge and End Devices},'' \emph{Proceedings -
  International Conference on Distributed Computing Systems}, pp. 328--339,
  2017.

\bibitem{GoogleDistWired}
\BIBentryALTinterwordspacing
J.~Dean, G.~S. Corrado, R.~Monga, K.~Chen, M.~Devin, Q.~V. Le, M.~Z. Mao, M.~A.
  Ranzato, A.~Senior, P.~Tucker, K.~Yang, and A.~Y. Ng,
  ``{Large{\_}Deep{\_}Networks{\_}Nips2012},'' pp. 1--11, 2012. [Online].
  Available:
  \url{papers3://publication/uuid/3BA537A9-BC9A-4604-95A4-E6F64F6691E7}
\BIBentrySTDinterwordspacing

\bibitem{Tuor_01_AdaptiveControl}
\BIBentryALTinterwordspacing
S.~Wang, T.~Tuor, T.~Salonidis, K.~K. Leung, C.~Makaya, T.~He, and K.~Chan,
  ``{When Edge Meets Learning : Adaptive Control for Resource-Constrained
  Distributed Machine Learning},'' in \emph{INFOCOM}, 2018. [Online].
  Available:
  \url{https://researcher.watson.ibm.com/researcher/files/us-wangshiq/SW{\_}INFOCOM2018.pdf}
\BIBentrySTDinterwordspacing

\bibitem{Tuor_02_2018_DemoAbstract}
T.~Tuor, S.~Wang, T.~Salonidis, B.~J. Ko, and K.~K. Leung, ``{Demo abstract:
  Distributed machine learning at resource-limited edge nodes},'' \emph{INFOCOM
  2018 - IEEE Conference on Computer Communications Workshops}, pp. 1--2, 2018.

\bibitem{SGDisBetter}
\BIBentryALTinterwordspacing
L.~Bottou and O.~Bousquet, ``{The Tradeoffs of Large Scale Learning},'' in
  \emph{Advances in Neural Information Processing Systems}, J.~C. Platt,
  D.~Koller, Y.~Singer, and S.~Roweis, Eds.\hskip 1em plus 0.5em minus
  0.4em\relax NIPS Foundation (http://books.nips.cc), 2008, vol.~20, pp.
  161--168. [Online]. Available:
  \url{http://leon.bottou.org/papers/bottou-bousquet-2008}
\BIBentrySTDinterwordspacing

\bibitem{StalenessAwarePaper}
Z.~Wei, S.~Gupta, X.~Lian, and J.~Liu, ``{Staleness-Aware Async-SGD for
  distributed deep learning},'' \emph{IJCAI International Joint Conference on
  Artificial Intelligence}, vol. 2016-Janua, pp. 2350--2356, 2016.

\bibitem{QIP_NP_ArXiV}
A.~D. Pia, S.~S. Dey, and M.~Molinaro, ``{Mixed-integer Quadratic Programming
  is in NP},'' pp. 1--10, 2014.

\bibitem{OPTI_CW12a}
J.~Currie and D.~I. Wilson, ``{OPTI: Lowering the Barrier Between Open Source
  Optimizers and the Industrial MATLAB User},'' in \emph{Foundations of
  Computer-Aided Process Operations}, N.~Sahinidis and J.~Pinto, Eds.,
  Savannah, Georgia, USA, 2012.

\bibitem{MNIST_IEEE}
Y.~LeCun, L.~Bottou, Y.~Bengio, and P.~Haffner, ``{Gradient-based Learning
  Applied to Document Recognition},'' \emph{Proceedings of IEEE}, vol.~86,
  no.~11, 1998.

\bibitem{NNComplexity}
\BIBentryALTinterwordspacing
{Brendan Shillingford}, ``{What is the time complexity of backpropagation
  algorithm for training artificial neural networks? - Quora},'' 2016.
  [Online]. Available:
  \url{https://www.quora.com/What-is-the-time-complexity-of-backpropagation-algorithm-for-training-artificial-neural-networks}
\BIBentrySTDinterwordspacing

\bibitem{WiFiChannelModel}
S.~Cebula, A.~Ahmad, J.~M. Graham, C.~V. Hinds, L.~A. Wahsheh, A.~T. Williams,
  and S.~J. DeLoatch, ``{Empirical channel model for 2.4 GHz ieee 802.11
  wlan},'' \emph{Proceedings of the 2011 International Conference on Wireless
  Networks}, 2011.

\end{thebibliography}

\appendices

\section{Proof of Theorem \ref{theorem1}}
From the KKT optimality conditions, we have the following relations given by the following equations:
\begin{equation}
    C_k^2  \tau d_k + C_k^1 d_k + C_k^0 - T \leq 0, \quad k=1,\ldots,K
    \label{Eq_20_KKTfirst}
\end{equation}
\begin{equation}
    \Gamma \succeq 0
\end{equation}
\begin{equation}
    \lambda_k\left(C_k^2  \tau d_k + C_k^1 d_k + C_k^0 - T\right) = 0, \quad k=1,\ldots,K
    \label{Eq_23_LambdaConstraints}
\end{equation}
\begin{equation}
    \nu_i (\sum_{k = 1}^{K}d_k - d) = 0     \quad i=1,2
    \label{Eq_24_EqualityConstraints}
\end{equation}
\begin{equation}
    \alpha_0 \tau = 0
\end{equation}
\begin{equation}
    \alpha_k d_k = 0   \quad k=1,\ldots,K
\end{equation}
\begin{multline}
     -\nabla\tau + \sum_{k = 1}^{K} \lambda_k\nabla\left(C_k^2  \tau d_k + C_k^1 d_k + C_k^0 - T\right) + \\ \nu_1\nabla\left(\sum_{k = 1}^{K}d_k - d\right) - \nu_2\nabla\left(\sum_{k = 1}^{K}d_k - d\right) - \\ \alpha_0\nabla\tau - \alpha_k\nabla\left(\sum_{k = 1}^{K}d_k\right) = 0 \label{Eq_20_KKTlast}      
\end{multline}
     



From the conditions in (\ref{Eq_20_KKTfirst}), we can see that the batch size at user $k$ must satisfy (\ref{Eq_30_BoundOnDk}). Moreover, it can be inferred from (\ref{Eq_23_LambdaConstraints}) that the bound in (\ref{Eq_30_BoundOnDk}) holds with equality for $k$ having  $\lambda_k \neq 0$.

In addition, it is clear from (\ref{Eq_24_EqualityConstraints}) that either $\nu_1 = \nu_2 = 0$ (which means that there is no feasible MEL solution and the orchestrator must offload the entire task to the edge/cloud servers) or $\sum_{k = 1}^{K}d_k = d$ (which we get if the problem is feasible). By re-writing the bound on $d_k$ in (\ref{Eq_30_BoundOnDk}) as an equality and taking the sum over all $k$, we have the following relation:
\begin{equation}
    d= \sum_{k=1}^{K} d_k^* = \sum_{k=1}^{K}  \left[ \dfrac{T-C_k^0}{\tau^* C_k^2+C_k^1} \right] = \sum_{k=1}^{K}  \left[ \dfrac{a_k}{\tau^* + b_k} \right]
    \label{Eq_37_BoundOnDkSums}
\end{equation}
The expression on the right-most hand-side has the form of a partial fraction expansion of a rational polynomial function of $\tau$. Therefore, we can expand it in the following way:
\begin{multline}
    \dfrac{a_1}{\tau + b_1} + \dfrac{a_2}{\tau + b_2}  + \dots + \dfrac{a_k}{\tau + b_k} + \dots +
    \dfrac{a_K}{\tau + b_K} = \\ \dfrac{1}{(\tau + b_1)(\tau + b_2)\dots(\tau + b_k)\dots(\tau + b_K)} \times \\ \Bigg[ a_1(\tau + b_2)(\tau + b_3)\dots(\tau + b_k)\dots(\tau + b_K) + ~  \\  a_2(\tau + b_1)(\tau + b_3)\dots(\tau + b_k)\dots(\tau + b_K) + ... +  \\ . a_k(\tau + b_1)(\tau + b_2)\dots(\tau + b_{k-1})(\tau + b_{k+1})\dots(\tau + b_K)  \\  + ... +  a_K(\tau + b_1)(\tau + b_2)\dots(\tau + b_{K-1}) \Bigg]
\end{multline}
Finally, the expanded form can be cleaned up in the form of a rational function with respect to $\tau$, which is equal to the total dataset size $d$.  
\begin{equation}
    d = \dfrac{\sum_{k=1}^{K} a_k \prod_{\substack{l=1 \\ l\neq k}}^{K} \left(\tau^*+b_l\right)}{\prod_{k=1}^{K} \left(\tau^*+b_k\right)}
\label{Eq_39_TauBoundsPolynomialForm}    
\end{equation}
Please note that the degrees of the numerator and denominator will be $K-1$ and $K$, respectively. Furthermore, the poles of the system will be $-b_k$, and, since $b_k \geq 0$, the system will be stable. Furthermore, $\tau = -b_k$ is not a feasible solution for the problem, because it is eliminated by the $\tau \geq 0$ constraint. Therefore, we can re-write (\ref{Eq_39_TauBoundsPolynomialForm}) as shown in (\ref{Eq_36_AnalBound}). By solving this polynomial, we obtain a set of solutions for $\tau$, one of them is feasible. The problem being non-convex, this feasible solution $\tau^*$ will constitute the upper bound to the solution of the relaxed problem. 


\end{document}